\documentclass[journal,10pt]{IEEEtran}
\IEEEoverridecommandlockouts
\usepackage{epsfig}
\usepackage{cite}
\usepackage{amsfonts}
\usepackage{graphics}
\usepackage{amsmath}
\usepackage{amsthm}
\usepackage{amssymb}
\usepackage{latexsym}
\usepackage[]{fontenc}
\usepackage{psfrag}
\usepackage{dsfont}
\usepackage{booktabs}
\usepackage{bm}
\usepackage{algorithm}
\usepackage{algpseudocode}
\usepackage{algpascal}
\usepackage{color}
\usepackage{mathabx}
\newtheorem{lemma}{Lemma}
\newtheorem{theorem}{Theorem}
\newtheorem{proposition}{Proposition}
\newtheorem{corollary}{Corollary}
\def\bg{{\boldsymbol{g}}}

\def\bn{{\boldsymbol{n}}}

\def\bw{{\boldsymbol{w}}}

\def\by{{\boldsymbol{y}}}

\def\bG{{\boldsymbol{G}}}

\def\bI{{\boldsymbol{I}}}

\def\bR{{\boldsymbol{R}}}

\begin{document}
\title{Optimal Pilot and Payload Power Control in Single-Cell Massive MIMO Systems}

\author{\IEEEauthorblockN{Hei Victor Cheng, Emil Bj\"ornson, and Erik G. Larsson}

\IEEEauthorblockA{Department of Electrical Engineering (ISY), Link\"oping University, Sweden\\
Email: \{hei.cheng, emil.bjornson, erik.g.larsson\}@liu.se}

\thanks{This work was supported by ELLIIT, the Link\"oping University Center for Industrial Information Technology (CENIIT), and the EU FP7 Massive MIMO for Efficient Transmission (MAMMOET) project.}}

\maketitle

\begin{abstract}
 This paper considers the jointly optimal pilot and data power allocation in single-cell uplink massive multiple-input-multiple-output (MIMO) systems. Using the spectral efficiency (SE) as performance metric and setting a total energy budget per coherence interval, the power control is formulated as optimization problems for two different objective functions: the weighted minimum SE among the users and the weighted sum SE. A closed form solution for the optimal length of the pilot sequence is derived. The optimal power control policy for the former problem is found by solving a simple equation with a single variable. Utilizing the special structure arising from imperfect channel estimation, a convex reformulation is found to solve the latter problem to global optimality in polynomial time. The gain of the optimal joint power control is theoretically justified, and is proved to be large in the low SNR regime. Simulation results also show the advantage of optimizing the power control over both pilot and data power, as compared to the cases of using full power and of only optimizing the data powers as done in previous work.
\end{abstract}

\begin{IEEEkeywords}
massive MIMO, power control, power allocation, convex optimization
\end{IEEEkeywords}
\section{Introduction}
\subsection{Background and Motivation}
Massive MIMO communication systems have recently attracted a lot of
attention\cite{RPLLMET2013,TM2010,LLSAZ2014}. The idea of massive MIMO
is to use a large amount of antennas at the base station (BS) to serve
multiple users in the same time and frequency resource block. The
ability to increase both SE and energy efficiency makes it one of the
key technologies for the 5G cellular networks. The performance
analysis of massive MIMO is of vast importance and has been done
in \cite{NLM2013,YM2013a} for single-cell systems and
in \cite{YM2013,HBD2013} for multi-cell systems. However the analysis
has been done with the assumption of equal or arbitrary fixed power
allocation among the users. Several previous papers
  \cite{YM2014,GFGFA2014,NML2014,C2014,CBL2015,LBL2015,GGA2015,ZZ2016,LW2015}
  have dealt with power control and provided initial results. (For
  relation to our work, see below.) In order to harvest all the
benefits brought by the massive antenna arrays and guarantee certain
uplink system performance, power control among the users is
necessary. This can be done by varying the power of different users to
increase the sum SE, provide services with certain fairness, or
balance between these goals.

Power control in wireless networks has been an important problem for decades, dating back to single-antenna wireless systems. Due to the interference from other users the power control is usually hard to solve optimally, in particular NP-hardness was proven in \cite{LZ2008} for the objective of maximizing the sum performance in single-antenna wireless networks, even with single-carrier transmission. For practical use a reasonable approach is to develop suboptimal algorithms with affordable complexity while achieving an acceptable performance, as done for example in \cite{CTPNJ2007}.

Compared to power control in single-antenna systems, power control in massive MIMO networks is a relatively new topic. Accurate channel estimates are needed at the BS for carrying out coherent linear processing, e.g. uplink detection and downlink precoding. Due to the large number of antennas in massive MIMO the instantaneous channel knowledge, which is commonly assumed to be known perfectly in the power control literature, is hard to obtain perfectly. The literature on power control for multi-user MIMO, and even jointly with optimal beamformer design, see for example\cite{SSB2008,BJ2013} and the references therein, did not consider the channel estimation error explicitly and the design criterion was based on SE. We want to provide power control schemes that optimize the ergodic SE based on only the large-scale fading to simplify system design, and take into account the channel estimation errors. Therefore in this work we develop a new framework for power control that matches practical systems (i.e., ergodic SE and imperfect CSI), as the methods developed in the literature cannot be applied directly for massive MIMO systems.

\subsection{Related Work and Our Contributions}
Uplink pilots are used to estimate the uplink channels. One needs to take into account both the pilot power and payload power, and hence optimal power control becomes even harder in massive MIMO compared to optimizing data power only in the single-antenna systems. Several work has tried to tackle this challenging problem. In \cite{YM2014} the authors optimize the data power for providing uniform service in multi-cell massive MIMO systems. In \cite{NML2014} and \cite{LW2015} the authors optimize the ratio between pilot and data power to maximize the sum SE, however each user is assumed to use the same ratio. In \cite{C2014} the sum data power is minimized subject to target signal to interference plus noise ratio (SINR) constraints for multi-cell massive MIMO systems. In \cite{GFGFA2014} power control is done to minimize the uplink power consumption under target SINR constraints where the authors optimize the pilot and data power iteratively to achieve local optima. In \cite{LBL2015} data power control is done to maximize various objectives in multi-cell massive MIMO systems with an iterative approach which only achieve local optimal solutions. In \cite{GGA2015} joint pilot and data power control is done to maximize the energy efficiency. The optimization is done with an approximation of the interference term which therefore does not give the optimal solution. In the conference version of this work\cite{CBL2015} we provided a GP formulation for joint pilot and data power control in single cell massive MIMO systems with MRC. Then it was re-derived in \cite{ZZ2016} to minimize the total power consumption while meeting target uplink and downlink SINR for the users.
We are not aware of any work except \cite{CBL2015} and \cite{ZZ2016} that find the jointly global optimal pilot and payload data power for massive MIMO. The previous work either focus on power minimization with SINR constraints or only achieve local optima. The preset target SINR constraints are hard to obtain in practice and the local optima does not provide the complete information about how much can we gain by power control. In this work we address this by providing globally optimal joint power for various objectives, and the questions we want to answer are:
\begin{enumerate}
\item Is power control on the pilots needed for massive MIMO systems? If the answer is yes, how much can we gain from jointly optimizing the pilot power and data power, as compared to always using equal power allocation or just power control over the data power?
\item In which scenarios can we gain the most from joint optimization?
\item What intuition can be obtained from the optimal power control? This includes the pilot length, and how the pilot and payload power depend on the estimation quality and signal to noise ratio (SNR).
\end{enumerate}
In this paper we provide answers to these questions in the single-cell uplink scenario with linear operation including maximum ratio combining (MRC) and zero-forcing (ZF). The single-cell scenario is considered here to gain some initial insights to the problem and the challenging extension to multi-cell is left for future work. Note that there are important scenarios when single-cell massive MIMO systems can be deployed, e.g. stadiums and rural wireless broadband access. We formulate and solve the optimization problems and compare the results with simple heuristic power control policies. Two commonly used performance objectives, namely weighted max-min SE and weighted sum SE optimization, are investigated. Our contributions are the following:
\begin{enumerate}
\item For the weighted max-min SE formulation, a semi-closed form solution is obtained by solving a simple equation with a single variable.
\item For the weighted sum SE formulation, which was proved to be NP-hard in general wireless networks, is transformed into a convex form in the massive MIMO setup where efficient polynomial time algorithms can be applied to find the global optimum.
\item Both theoretical and numerical results are presented to show the gains of the new framework for joint pilot and data power control.
\end{enumerate}
The existing literature on power control is summarized in Table \ref{table}, the ones marked with ${*}$ are the contributions in this work.

\begin{table*}[t]
\caption{\label{table}Existing Methods for Power Control Problem}
\begin{center}
    \begin{tabular}{ | l | l | l | p{4cm} |}
    \hline
    Problems & Massive MIMO (data) & Massive MIMO (data+pilot) & Multiuser MIMO (perfect CSI) \\ \hline
    Max-min (MRC) & closed form \cite{YM2014} & {*}semi-closed form (Theorem \ref{snr}) & convex \cite{BJ2013} \\ \hline
    Max-min (ZF) & closed form \cite{YM2014} & {*}semi-closed form (Theorem \ref{snr})& full power\\ \hline
    Sum (MRC) & {*}virtual water-filling (Algorithm \ref{vwf}) & {*}convex (Theorem \ref{summrc})& NP-hard \cite{LZ2008}\\ \hline
    Sum (ZF) & {*}virtual water-filling (Algorithm \ref{vwf}) &{*}convex (Corollary \ref{sumzf})& full power \\    \hline
    \end{tabular}
\end{center}
\end{table*}

The rest of the paper is organized as follows. Section \ref{model}
presents the system model and all the necessary notations. Lower
bounds on the uplink capacity are presented which are used to define
the problem formulations for optimal power control. In Section
\ref{training} we obtain the optimal pilot length for both problem
formulations. Section \ref{maxminproblem} derives the solution
approaches for solving the power control problems with weighted
max-min SE. In Section \ref{sumproblem} the weighted sum SE
formulation is studied. Section \ref{correlated}
  discusses the extension of the methodology developed in this paper
  to correlated channel fading models. In Section \ref{simulation}
simulation results and discussion of the results are
presented. Finally in Section \ref{conclusion} we draw some
conclusions.

\section{System Model}\label{model}
Consider an uplink single-cell massive MIMO systems with $M$ antennas at the BS and $K$ single-antenna users. The $K$ users are assigned $K$ orthogonal pilot sequences of length $\tau_p$ for $K \leq \tau_p \leq T$, where $T$ is the number of symbols in the coherence interval in which the channels are assumed to be constant. The channels are modeled to be independent Rayleigh fading as this matches the non-line-of-sight massive MIMO channel measurement results reported in\cite{GERT2015}. The flat fading channel matrix between the BS and the users is denoted by $\bG\in \mathbb{C}^{M \times K}$, where the $k^{th}$ column represents the channel response to user $k$ and has the distribution
\begin{equation}
\bg_{k}\sim CN(\mathbf{0},\beta_k\mathbf{I}), ~k=1,2,\ldots, K,
\end{equation}
which is a circularly symmetric complex Gaussian random vector. The variance $\beta_k>0$ represents the large-scale fading including path loss and shadowing, and is normalized by the noise variance at the BS to simplify the notation. The large-scale fading coefficients are assumed to be known at the BS as they are varying slowly (in the scale of thousands of coherence intervals) and can be easily estimated. The power control proposed in this work only depends on the large-scale fading which makes it feasible to optimize the power control online.

In each coherence interval, user $k$ transmits its orthogonal pilot sequence with power $p_p^k$ to enable channel estimation at the BS. We assume that minimum mean-squared error (MMSE) channel estimation is carried out at the BS to obtain the small-scale coefficients. This gives an MMSE estimate of the channel vector from user $k$ as
\begin{equation}
\hat{\bg}_k=\frac{\sqrt{\tau_p p_p^k \beta_k}}{1+\tau_p p_p^k \beta_k}\left(\sqrt{\tau_p p_p^k}\bg_k+\bn_p^k\right)
\end{equation}
where $\bn_p^k\sim CN(\bf{0},\bI)$ accounts for the additive noise during the training interval. During the payload data transmission interval, the BS receive the signal
\begin{equation}
\by=\sum_{k=1}^K \bg_k \sqrt{p_d^k}s_k+\bn
\end{equation}
where $s_k$ is the zero mean and unit variance Gaussian information symbol from user $k$ and $\bn\sim CN(\bf{0},\bI)$ represents the noise during the data transmission. The channel estimates are used for MRC or ZF detection of the payload, which corresponds to multiplying the received signal $\by$ with $\hat{\bG}^H\triangleq[\hat{\bg}_1,\ldots,\hat{\bg}_K]^H$ or $(\hat{\bG}^H \hat{\bG})^{-1}\hat{\bG}^H$ to detect the symbols $s_1,\dots, s_K$. The power control methodologies derived in this paper can be applied jointly to each subcarrier in an orthogonal frequency division multiplexing (OFDM) systems. With the channel hardening effect offered by massive MIMO, channel variations in different subcarriers can be neglected and the SE in every subcarrier will mainly depend on the large-scale fading. Therefore the whole spectrum can be allocated to every user and the same power control can be applied to all subcarriers. To make a fair comparison with the scheme with equal power allocation in which each user gives the same power to pilot and data, as done in \cite{NLM2013} and most other previous work, we impose the following constraint on the total transmit energy over a coherence interval:
\begin{equation}\label{power}
\tau_p p_{p}^k+(T-\tau_p) p_d^k \leq E_{k},~k=1,\ldots,K
\end{equation}
where $E_{k}$ is the total energy budget for user $k$ within one coherence interval. In previous work, $p_{p}^k$ and $p_{d}^k$ have been optimized separately or often not optimized at all in which case the massive MIMO ability to provide high SE for each user cannot be fully harvested. Therefore we consider the scenario where each user can choose freely how to allocate its energy budget on the pilots and payload. In \cite{NML2014,HBD2013} $p_p^k$ and $p_d^k$ are set equal for every user. The work \cite{YM2014,C2014,LBL2015} optimized the payload power to maximize the minimum throughput, which corresponds to fixing $p_p^k$ for every user and optimizing only over $p_d^k$. The work \cite{BLD2015} adopted inverse power control for the pilot power, which corresponds to setting $p_p^k=C/\beta_k$ with a normalization constant $C$ and the data power $p_d^k$ are set to be equal for all users. These previous work can all be included in our framework by setting different variables to be constant. Therefore our framework of power control is the most general so far.

\subsection{Achievable SE With Linear Detection}\label{rate}
Since the exact ergodic capacity of the uplink multiuser channels with channel uncertainty is unknown, lower bounds on the achievable SE are often adopted as the performance metric in the massive MIMO literature. Here we present lower bounds on the capacity for arbitrary power control. The achievable SE for user $k$ using MRC is given by the following lemma.

\begin{lemma}\label{mrcse}
The capacity of user $k$ with MRC detection is lower bounded by the achievable ergodic SE
\begin{equation}
R_k=\left(1-\frac{\tau_p}{T}\right)\log_2(1+\mathrm{SINR}_k)
\end{equation}
where pilot and payload powers are arbitrary,
\begin{equation}\label{sinrmrc}
\mathrm{SINR}_k=\frac{Mp_d^k \gamma_k }{1+\sum_{j= 1}^K \beta_j p_d^j}
\end{equation}
and $\gamma_k=\frac{\tau_p p_p^k \beta_k^2}{1+\tau_p p_p^k \beta_k}$.
\end{lemma}

For ZF, an achievable ergodic SE of user $k$ is given by the following lemma.
\begin{lemma}\label{zfse}
The capacity of user $k$ with ZF detection is lower bounded by the achievable ergodic SE
\begin{equation}
R_k=\left(1-\frac{\tau_p}{T}\right)\log_2(1+\mathrm{SINR}_k)
\end{equation}
where pilot and payload powers are arbitrary,
\begin{equation}\label{sinrzf}
\mathrm{SINR}_k=\frac{(M-K)p_d^k \gamma_k}{1+\sum_{j=1}^K  p_d^j(\beta_j-\gamma_j)}
\end{equation}
and $\gamma_j=\frac{\tau_p p_p^j \beta_j^2}{1+\tau_p p_p^j \beta_j}$. $M>K$ needs to be satisfied for ZF detector to work.
\end{lemma}

\color{black}
The proofs of Lemmas \ref{mrcse} and \ref{zfse} can be obtained by
adding corresponding indices for different users' pilot power in the proofs of \cite{YM2013}. Note that these achievable rates are valid for any number of antennas at the BS. However, they are only close to the capacity when there is substantial channel hardening, which is the case when $M$ is large, i.e. in the massive MIMO regime.

These achievable SEs are the performance metric commonly used in the massive MIMO literature. Therefore it is used throughout the paper, where $\tau_p$, $p_p^k$ and $p_d^k$ are the variables to be optimized (for $k=1,\ldots,K$). The optimization can be done at the BS, which can then inform the users about the pilot length, the amount of power to be spent on pilots, and the amount of power to be spent on payload data. The aim is to maximize a given utility function $U(R_1,\ldots,R_K)$ where $U(\cdot)$ can be any function that is monotonically increasing in every argument. The utility function characterizes the performance and fairness that we provide to the users. Examples of commonly used utility functions are the max-min fairness, sum performance, and proportional fairness\cite{BJ2013}. The general problem we address for both MRC and ZF is:
\textbf{\begin{equation}\label{opt}
\begin{aligned}
& ~\underset{\tau_p,\{p_p^k\},\{p_d^k\}}{\text{maximize}}
& & U ~(R_1,\ldots,R_K)   \\
&~\text{subject to}
& & \tau_p p_{p}^k+(T-\tau_p) p_d^k \leq E_{k}, ~\forall k, \\
& & & p_p^k\geq 0, p_d^k\geq0,~ \forall k, \\
& & & K\leq\tau_p\leq T.
\end{aligned}
\end{equation}}
\section{Optimal Pilot Length}\label{training}
In this section we derive the optimal length of the pilot sequences in \eqref{opt} in closed form. First we provide the following lemma:
\begin{lemma}\label{lemma1}
For any monotonically increasing utility function with MRC or ZF detection, the energy constraint \eqref{power} is satisfied with equality for every user at the optimal solution, i.e.,
\begin{equation}
\tau_p p_{p}^k+(T-\tau_p) p_d^k = E_{k}, ~k=1,\ldots,K
\end{equation}
at the optimal point of \eqref{opt}.
\end{lemma}
\begin{proof}
We prove this by contradiction. The SINRs in \eqref{sinrmrc} and \eqref{sinrzf} for MRC and ZF are monotonically increasing in $p_p^k$ for every user $k$, and independent of the other users' pilot powers. Suppose some users do not use the full energy budget in the optimal power allocation, they can each increase their pilot power to improve their own SINR without lowering any other user's SINR. Therefore we create a solution which is better than or equal to the optimal one, which is a contradiction to our assumption. Therefore the energy constraint is satisfied with equality.
\end{proof}
Then we state the following theorem which gives the optimal length of training interval in closed form.
\begin{theorem}\label{theorem1}
For any monotonically increasing utility function $U(R_1,\ldots,R_k)$, the problem \eqref{opt} has $\tau_p=K$ at the optimal solution.
\end{theorem}
\begin{proof}
The proof can be found in Appendix \ref{proofpilot}.
\end{proof}
Using Theorem \ref{theorem1}, we can reduce the number of variables involved in \eqref{opt} and this enables us to find the optimal solutions for certain utility functions in the following sections. Also from Theorem \ref{theorem1} we know that the optimal training period $\tau_p$ is equal to the number of users being served, and is the same for every user. Therefore there is no need for assigning pilot sequences of different lengths to different users.
\section{Joint Power Control of Pilots and Payload to Maximize Weighted Minimum SE}\label{maxminproblem}

In this section we solve the power control problem \eqref{opt} for the
class of max-min fairness problem. The max-min fairness problem is selected to provide the same quality-of-service to all users in the cell. The two cases with MRC and ZF will be discussed separately since the SINR expressions are different. With
max-min fairness we aim at serving every user with equal weighted SE
according to their priorities and make this value as large as
possible. We choose $U(\tilde{R}_1,\ldots,\tilde{R}_K)=\min_k
\tilde{R}_k$ with
$\tilde{R}_k=(1-\frac{\tau_p}{T})\log_2(1+w_k\text{SINR}_k)$ where
$w_k>0$ are weighting factors to prioritize different users and enable us to achieve any point on the Pareto
  boundary of the achievable rate region $(R_1,\ldots,R_K)$ by varying
  the weights \cite{BJDO2014}. It is trivial to extend Theorem
\ref{theorem1} to this case and prove that the optimal length of
training equal to $K$. Since
$(1-\frac{\tau_p}{T})\log_2(1+w_k\text{SINR}_k)$ is monotonically
increasing in $w_k\text{SINR}_k$, it is equivalent to choose objective
as $\min_k w_k\text{SINR}_k$.
\subsection{Max-Min for MRC}

With MRC, the power control problem becomes
\begin{equation}\label{maxminsinr}
\begin{aligned}
& \underset{\{p_p^k\},~\{p_d^k\}}{\text{maximize}}
& & \underset{k}{\text{min}} ~\frac{w_kMp_d^k \gamma_k }{1+\sum_{j= 1}^K \beta_j p_d^j}   \\
& \text{subject to}
& & \tau_p p_{p}^k+(T-\tau_p) p_d^k \leq E_{k}, \forall k \\
& & & p_p^k\geq 0, p_d^k\geq0, \forall k.
\end{aligned}
\end{equation}

\subsubsection{Geometric Program Formulation}

Using the epigraph form of \eqref{maxminsinr} we have the following equivalent problem formulation:
\begin{equation}\label{maxmin}
\begin{aligned}
& \underset{\{p_p^k\},\{p_d^k\},~\lambda}{\text{maximize}}
& &  \lambda  \\
& ~\text{subject to}
& & w_kMp_d^k \tau_pp_p^k \beta_k^2 \geq \\
&&& \lambda(1+\sum_{j=1}^K \beta_j p_d^j +\tau_p p_p^k\beta_k \\
&&&+ \tau_p p_p^k \beta_k \sum_{j=1}^K \beta_j p_d^j), \forall k\\
& & & \tau_p p_{p}^k+(T-\tau_p) p_d^k \leq E_{k}, \forall k \\
& & & p_p^k\geq 0, p_d^k\geq0, \forall k.
\end{aligned}
\end{equation}

This problem is non-convex as it is formulated here, however we recognize it as a geometric program (GP). The GP formulation has been considered in the conference version of this paper\cite{CBL2015}. Since we next present a new semi-closed form solution with much lower complexity, the GP details are omitted here and we refer the interested readers to \cite{CBL2015}.

\subsubsection{Explicit Solution}
Next we develop a semi-closed form solution to the max-min fairness problem. Before we present the solution, we need the following lemma:

\begin{lemma}\label{equalsnr}
At the optimal point, all $w_kp_d^k\gamma_k$ are equal, i.e.,
\begin{equation}\label{equalsnr2}
w_kp_d^k\gamma_k=w_jp_d^j\gamma_j, ~\forall ~j,k=1,\ldots,K.
\end{equation}
\end{lemma}

\begin{proof} First we need the key observation that at the optimal solution, all weighted $\text{SINR}_k$ are equal. We prove this by contradiction. Assume that at the optimal solution, there is at least one user $k$ that has a higher weighted SINR than the others. Denote the minimum weighted SINR at the optimal solution as $\text{SINR}^*$. We can then construct a new solution by decreasing $p_d^k$ by $\delta>1$ while maintaining that $w_k\text{SINR}_k>\text{SINR}^*$. Since $w_k\text{SINR}_k$ is a continuous increasing function in $p_d^k$, we can always find such $\delta>1$. Keeping the other users' powers fixed, we have increased all other users' weighted SINRs. Then we have $w_j\text{SINR}_j>\text{SINR}^*,~\forall j$, hence we constructed a solution that is better than the optimal solution, which is a contradiction to the initial assumption. Therefore at the optimal solution all weighted $\text{SINR}_k$ are equal, and we have
\begin{equation}\label{equalsinr}
w_k\text{SINR}_k=\frac{w_kMp_d^k \gamma_k}{1+\sum_{j=1}^K \beta_j p_d^j}=\widebar{\text{SINR}}, \forall k=1,\ldots,K,
\end{equation}
where $\widebar{\text{SINR}}$ is the common weighted SINR for every user.
We observe that the denominator is the same for every user $k$. Therefore the numerator of \eqref{equalsinr} is the same for all $k$, which leads to \eqref{equalsnr2}.
\end{proof}
We call $w_kp_d^k\gamma_k$ the weighted receive signal power $k$ ($\text{SP}_k$).
Then we want to find the $p_d^k$ that satisfies Lemma \ref{lemma1} for any given value of $x=w_kp_d^k\gamma_k$, which is provided in the following proposition:
\begin{proposition}\label{pusolution}
For any given value of the weighted $\text{SP}_k~w_kp_d^k\gamma_k=x$, the optimal $p_d^k$ is given in \eqref{root} on top of next page. When \eqref{root} is not real-valued, then such $\text{SP}_k$ is not attainable by any feasible power allocation.
\end{proposition}
\begin{figure*}[!t]
\begin{equation}\label{root}
p_d^k=\frac{E_{k}\beta_k +(T-K)\frac{x}{w_k}-\sqrt{E_{k}^2\beta_k^2-2(T-K)(E_{k}\beta_k+2)\frac{x}{w_k}+(T-K)^2(\frac{x}{w_k})^2}}{2(T-K)\beta_k}
\end{equation}
\hrulefill
\end{figure*}
\begin{proof} Making use of Lemma \ref{lemma1} and Theorem \ref{theorem1}, we have the following equation:
\begin{equation}
\frac{p_d^k(E_{k}-(T-K)p_d^k)\beta_k^2}{1+(E_{k}-(T-K)p_d^k)\beta_k}=\frac{x}{w_k}.
\end{equation}
This is equivalent to the quadratic equation
\begin{equation}
\begin{aligned}
&(T-K)\beta_k^2(p_d^k)^2-\beta_k\left(E_{k}\beta_k+\frac{(T-K)x}{w_k}\right)p_d^k\\
&+\frac{(E_{k}\beta_k+1)x}{w_k}=0.
\end{aligned}
\end{equation}
If the equation has real-valued roots, we observe that sum of roots and products of roots are positive, therefore both roots of the equation are positive. Inspecting \eqref{sinrmrc} we see that smaller $p_d^k$ gives a higher $\text{SINR}_k$ when $p_d^k\gamma_k$ is fixed. Therefore we arrive at the result. Moreover when the quadratic equation does not have real-valued roots, then $w_kp_d^k\gamma_k<x$ for all feasible $p_d^k$ and therefore such $\text{SP}_k$ is not attainable.
\end{proof}

We now reformulate Problem \eqref{maxminsinr} in terms of SP as presented in the following proposition:
\begin{proposition}
Problem \eqref{maxminsinr} is reduced to the optimization problem \eqref{maxmint} with one variable (given on top of next page),
where the optimization is done in the domain where the objective function is real, i.e., $x$ is constrained to be achievable for every user $k$. Finding the optimal $x$ in \eqref{maxmint} gives the optimal common SP for every user. By using Proposition \ref{pusolution} we can find the optimal $p_d^k$ and $p_p^k$ for every user $k$ to achieve this optimal common SP.
\end{proposition}
\begin{figure*}[!t]
\normalsize
\begin{equation}\label{maxmint}
\begin{aligned}
& \underset{x\geq0}{\text{minimize}}
& & \frac{1}{x}+\frac{E_{k}}{2(T-K)x}\sum_k \beta_k-\frac{\sum_k\sqrt{\frac{E_{k}^2\beta_k^2-2(T-K)(E_{k}\beta_k+2)\frac{x}{w_k}+(T-K)^2(\frac{x}{w_k})^2}{x^2}} }{2(T-K)},  \\
\end{aligned}
\end{equation}
\begin{equation}\label{maxminx}
\begin{aligned}
&\frac{1}{2(T-K)}\sum_k \frac{E_{k}^2\beta_k^2y-(T-K)(E_{k}\beta_k+2)\frac{1}{w_k}}{\sqrt{E_{k}^2 \beta_k^2 y^2-2(T-K)(E_{k}\beta_k+2)\frac{1}{w_k}y+(T-K)^2\frac{1}{w_k^2}}}=1+\frac{E_{k}}{2(T-K)}\sum_k \beta_k.
\end{aligned}
\end{equation}
\hrulefill
\end{figure*}

\begin{proof} We first define $x=w_kp_d^k\gamma_k$ and substitute the results from Proposition \ref{pusolution} into the expression of $\text{SINR}_k$. Then the objective function is obtained by changing the maximization of SINR to minimization of 1/SINR and simplifying the expression.
\end{proof}

Finally we present the solution to Problem \eqref{maxmint}:
\begin{theorem}\label{snr}
The common SP that maximizes the minimum weighted SINR is given by $1/y$ where $y$ is the unique optimal solution to an strictly convex optimization problem, and the unique real-valued solution can be found by solving the equation \eqref{maxminx} on top of next page.
\end{theorem}
\begin{proof}
First we make the change of variable $y=1/x$ in \eqref{maxmint}, then we have the following problem:
\begin{equation}\label{maxminy}
\begin{aligned}
& \underset{y}{\text{minimize}} ~~y+\frac{yE_{k}}{2(T-K)}\sum_k \beta_k- \\
&\frac{\sum_k\sqrt{E_{k}^2\beta_k^2y^2-2(T-K)(E_{k}\beta_k+2)\frac{y}{w_k}+(T-K)^2(\frac{1}{w_k})^2}}{2(T-K)}.  \\
\end{aligned}
\end{equation}

The first term is linear, thus the objective is convex if the last term, which has the form $f(y)=\sqrt{ay^2+by+c}$ is concave.  This is verified by taking the second derivative of $f(y)$ which gives
\begin{equation}
\frac{1}{4}\frac{4ac-b^2}{(ax^2+bx+c)^{3/2}}.
\end{equation}
The second derivative is non-positive when $b^2-4ac\geq0$, in such case $f(y)$ is concave.

The $k^{th}$ square root term in \eqref{maxminy} satisfies $b^2-4ac\geq0$ as
\begin{equation}
\begin{aligned}
&\left(2(T-K)(E_{k}\beta_k+2)\frac{1}{w_k}\right)^2-4E_{k}^2\beta_k^2(T-K)^2\left(\frac{1}{w_k}\right)^2\\
=&4(T-K)^2\left(\frac{1}{w_k}\right)^2(4E_{k}\beta_k+4)>0,
\end{aligned}
\end{equation}
and hence it is strictly concave. The overall function is thus strictly convex. Hence the optimal $y$ can be found by setting the first derivative of the objective to zero and the unique solution is found.
\end{proof}

Since we know that \eqref{maxminy} is a strictly convex function in $y$, hence there will be only one optimal solution and it can be found by line search, such as using bisection method, which makes it easy to implement.

To summarize, we provided a semi-closed form solution to the max-min SE problem with the following procedure:
\begin{enumerate}
\item Find the optimal common weighted SP by solving \eqref{maxminx} given in Theorem \ref{snr}, using e.g. bisection.
\item For this SP find all the optimal $p_d^k$ using Proposition \ref{pusolution}.
\item Find the optimal $p_p^k$ using Lemma \ref{lemma1}.
\end{enumerate}

Finding the optimal power control parameters is reduced to solving an equation with a single variable (or a single-variable convex problem). Therefore the complexity is linear in the number of users being served and independent of the number of antennas, which can be implemented in real-time at the BS.

\subsection{Max-Min for ZF}
Similar to the case of the MRC detector, we can write the problem as max-min weighted SINR as follows:
\begin{equation}\label{maxminsinrzf}
\begin{aligned}
& \underset{\{p_p^k\},~\{p_d^k\}}{\text{maximize}}
& & \underset{k}{\text{min}} ~\frac{w_k(M-K)p_d^k \gamma_k}{1+\sum_{j=1}^K  p_d^j(\beta_j-\gamma_j)}  \\
& \text{subject to}
& & \tau_p p_{p}^k+(T-\tau_p) p_d^k \leq E_{k}, \forall k \\
& & & p_p^k\geq 0, p_d^k\geq0, \forall k.
\end{aligned}
\end{equation}

The only difference from \eqref{maxminsinr} is the expressions of the $\text{SINR}$s, which is now taken from \eqref{sinrzf} by inserting $\tau_p=K$. Due to the negative terms appearing in the denominator of the SINR expressions, this problem cannot be directly transformed to a GP problem. Fortunately we observe that the denominators of the SINRs are the same for all users, therefore we can state a similar result as Lemma \ref{equalsnr}.
\begin{lemma}\label{equalsnrzf}
For the ZF detector, at the optimal point, all $w_kp_d^k\gamma_k$ are equal, i.e.,
\begin{equation}
w_kp_d^k\gamma_k=w_jp_d^j\gamma_j, ~\forall ~j,k=1,\ldots,K.
\end{equation}
\end{lemma}
The proof is similar to that of Lemma \ref{equalsnr} and is omitted.

By using Lemma \ref{equalsnrzf} we obtain the following important result:
\begin{theorem}\label{zfequalmrc}
Problem \eqref{maxminsinrzf} can be reformulated as
\begin{equation}\label{maxminsinrzfequ}
\begin{aligned}
& \underset{\{p_p^k\},~\{p_d^k\}}{\mathrm{maximize}}
& & \underset{k}{\mathrm{min}} ~ \frac{w_k(M-K)p_d^k \gamma_k}{1+\sum_{j=1}^K  p_d^j\beta_j}  \\
& \mathrm{subject ~to}
& & \tau_p p_{p}^k+(T-\tau_p) p_d^k \leq E_{k}, \forall k \\
& & & p_p^k\geq 0, p_d^k\geq0, \forall k.
\end{aligned}
\end{equation}
This implies that solving problem \eqref{maxminsinrzfequ} gives the same optimal ${p_d^k},{p_p^k}$ as solving problem \eqref{maxminsinrzf}, but the objective value is different.
\end{theorem}

\begin{proof}
Using Lemma \ref{equalsnrzf} we have $w_kp_d^k\gamma_k=w_jp_d^j\gamma_j$ at the optimal point. Moreover the denominator can be written as $1+\sum_{j=1}^K  p_d^j\beta_j-\sum_{j=1}^K p_d^j\gamma_j$ where the last term is equal to $w_kp_d^k\gamma_k\sum_{j=1}^K \frac{1}{w_j}$. Then we can rewrite the weighted SINR as
\begin{equation}
w_k\text{SINR}_k=\frac{M-K}{\frac{1}{w_kp_d^k\gamma_k}+\sum_j \frac{p_d^j\beta_j}{w_kp_d^k\gamma_k}-\sum_j\frac{1}{w_j}}
\end{equation}
Since $\sum_j\frac{1}{w_j}$ is a constant and the same for every user, the set of parameters that maximizes $\text{SINR}_k$ also maximizes the $\text{SINR}_k$ if the term $\sum_j\frac{1}{w_j}$ is removed. Therefore both problem are equivalent in the sense that they have the same optimal solutions.
\end{proof}

From Theorem \ref{zfequalmrc} we see that only the constant $M$ is replaced with $M-K$, therefore the power allocation that solves the weighted max-min SE for the MRC also solves the weighted max-min SE for the ZF. The same methods and analytical solutions apply. Therefore we don't need to do a separate optimization for ZF in this case. This implies that the users do not need to know what kind of detector is used at the BS. While the BS can switch between different detectors according to the data traffic requirements or power consumption restrictions.

\section{Joint Pilot and Data Power Control for Weighted Sum SE}\label{sumproblem}
In this section we solve the power control problem \eqref{opt} for the
weighted sum SE for MRC and ZF detector. This problem is selected to maximize the total system throughput, and weights are included to provide some fairness between different users. We
define the weighted sum SE by choosing $U(R_1,\ldots,R_K)=\sum_{k=1}^K
w_kR_k$.

Power control that maximizes sum SE when interference is present is known to be an $\mathrm{NP}$-hard problem in general under perfect channel knowledge\cite{LZ2008}. In this part we present a polynomial-time solution to one special case when all sources transmit to the same receiver. When channel estimation errors are present, with the bounding techniques we used for the achievable SE we discover a specific structure that lead to a convex reformulation after a series of transformations. Since optimizing the data power is considered to be a hard problem itself, in the following we first present the case when one only optimizes the data power, then the solution approach is extended to the case of joint optimization of pilot and data power.
\subsection{Weighted Sum SE for MRC}
By using Theorem \ref{theorem1}, \eqref{opt} now becomes the following optimization problem:
\begin{equation}\label{sumrate}
\begin{aligned}
& \underset{\{p_p^k\},~\{p_d^k\}}{\text{maximize}}
& &  \sum_k w_k\log_2\left(1+\frac{Mp_d^k \gamma_k }{1+\sum_{j= 1}^K \beta_j p_d^j}\right)   \\
& \text{subject to}
& & \tau_p p_{p}^k+(T-\tau_p) p_d^k \leq E_{k}, \forall k, \\
& & & p_p^k\geq 0, p_d^k\geq0, \forall k.
\end{aligned}
\end{equation}

\subsubsection{Optimizing Data Power}
In the case of optimizing data power only, the energy budget constraint reduced to the peak power constraint on the data power given as $P_k=(E_k-\tau_pp_p^k)/(T-\tau_p)$ for user $k$ where $p_p^k$ is now a constant. Therefore we have the following optimization problem:
\begin{equation}\label{sumratedata}
\begin{aligned}
& \underset{~\{p_d^k\}}{\text{maximize}}
& &  \sum_k w_k\log_2\left(1+\frac{Mp_d^k \gamma_k }{1+\sum_{j= 1}^K \beta_j p_d^j}\right)   \\
& \text{subject to}
& & p_d^k \leq P_{k}, \forall k, \\
& & & p_d^k\geq0, \forall k.
\end{aligned}
\end{equation}
In this case $\gamma_k$ are fixed constants and the optimization variables are the data power $p_d^k$ with individual power constraints. The formulation in \eqref{sumratedata} is non-convex. However, we use the observation that the denominator of the SINR expression is the same for every user, to obtain a convex reformulation as described in the following theorem:
\begin{theorem}\label{virtual_wf}
Problem \eqref{sumratedata} can be reformulated into the following convex form:
\begin{equation}\label{sumratedata2}
\begin{aligned}
& \underset{~s,~\{x_k\}}{\mathrm{maximize}}
& &  \sum_k w_k\log_2\left(1+a_k x_k \right)   \\
& \mathrm{subject~to}
& & x_k \leq \beta_kP_{k}s, \forall k, \\
& & & x_k\geq0, \forall k, \\
& & & \sum_{j= 1}^K x_j=1-s,
\end{aligned}
\end{equation}
where $a_k=M\gamma_k/\beta_k$.
The two formulations are equivalent in the sense that they have the same optimal objective value, and the solution to \eqref{sumratedata} can be obtained from solution to \eqref{sumratedata2} via $p_d^k=\frac{x_k}{s\beta_k}$.
\end{theorem}

\begin{proof} First we observe that the denominator of the SINR expression in the objective function of \eqref{sumratedata} is the same for every user. It is possible for us to apply the following variable substitutions:
\begin{enumerate}
\item $x_k=\frac{\beta_k p_d^k}{1+\sum_j \beta_j p_d^j}$;
\item $s=\frac{1}{1+\sum_j \beta_j p_d^j}$, or equivalently, $s=1-\sum_j x_j$.
\end{enumerate}
The individual power constraints are changed proportionally.
\end{proof}

Since problem \eqref{sumratedata2} is convex and Slater's condition is always satisfied, standard convex solvers can handle this problem. Moreover we observe  that Theorem \ref{virtual_wf} transforms the problem into a power allocation of virtual parallel channels with individual and sum power constraints. This problem has a water-filling structure when $s$ is fixed. Therefore we investigate the Karush-Kuhn-Tucker (KKT) conditions and obtain the following solution structure which enable us to develop dedicated algorithms that are more efficient than applying standard interior point methods. The results are summarized in the following theorem:

\begin{theorem}\label{v_wf}
The optimal power allocation to the virtual parallel channel \eqref{sumratedata2} satisfies the following equations:
\begin{enumerate}
\item $x_k=\min \left(\beta_k P_k s, \max \left(\frac{w_k}{\nu}-\frac{1}{a_k}\right)^{+}\right),~\forall k$,
\item $\sum_{j=1}^K x_j=1-s$,
\item $\nu=\sum_{j=1}^K \beta_j P_j \left(\frac{w_j}{\frac{1}{a_j}+x_j}-\nu\right)^{+}$,
\end{enumerate}
where $(z)^{+}=\max (z,0)$ for any real number $z$. When $s$ is fixed, the first two conditions are sufficient. Moreover, when
\begin{equation}\label{fullop}
\sum_{j=1}^K \beta_jP_j \left(\frac{w_j}{\frac{1}{a_j}+\frac{\beta_jP_j}{1+\sum_{j'} \beta_{j'} P_{j'}}}\right) \leq \min_k \frac{(1+\sum_j\beta_jP_j)w_k}{\frac{\beta_kP_k}{1+\sum_{j} \beta_{j} P_{j}}+\frac{1}{a_k}},
\end{equation}
then it is optimal to let every user use full power.
\end{theorem}

\begin{proof}
The proof can be found in Appendix \ref{proofvwf}.
\end{proof}

With Theorem \ref{v_wf} we develop an efficient algorithm to obtain the optimal power allocation. For fixed $s$ the optimal $x_k$ can be obtained via modified water-filling. Next we apply bisection on $s$ to find the optimal $s$ such that condition (3) in Theorem \ref{v_wf} is satisfied. The use of bisection needs to be justified and is also provided in the appendix. We only need to search for $s\in [\frac{1}{1+\sum_j\beta_jP_j},1]$ since this is an implicit constraint from the definition. The $s$ that solves the problem is such that $f(s)\triangleq \sum_{j=1}^K \beta_j P_j \left(\frac{1}{\frac{1}{a_j}+x_j}-v\right)^{+}-\nu=0$. As a by-product we also get the condition when it is optimal to for everyone to use full power. The procedure of finding the optimal power control parameters are described in Algorithm \ref{vwf}.

\alglanguage{pascal}
\begin{algorithm}
 \begin{algorithmic}[1]
 \State Initialize $s_l=\frac{1}{1+\sum_j\beta_jP_j}$ and $s_u=1$. Check if \eqref{fullop} is satisfied, if yes then terminate and output $p_d^k=P_k$.
 Otherwise compute $s=(s_l+s_u)/2$.
 \Repeat
 \State solve for $x_k$ and $v$ satisfying conditions (1) and (2) in Theorem \ref{v_wf}
 \State \textbf{if} $f(s)>0$
 \State $s_u=s$, $s_l$ remains unchanged
 \State $s\leftarrow(s_u+s_l)/2$
 \State \textbf{else}
 \State $s_l=s$, $s_u$ remains unchanged
 \State $s\leftarrow(s_u+s_l)/2$
 \Until \bf{convergence} with $|s_u-s_l|<\epsilon$
 \State \textbf{return} all $p_d^k=\frac{x_k}{\beta_ks}~\forall k$
 \end{algorithmic}
 \caption{Virtual Water-Filling Algorithm for \eqref{sumratedata2}}
 \label{vwf}
\end{algorithm}

\subsubsection{Joint Pilot and Data Power Optimization}
Next we extend the method to the case of joint power control over pilot and data power. The problem can be written as follows:
\begin{equation}\label{sumratejoint}
\begin{aligned}
& \underset{\{p_d^k\},~\{p_p^k\}}{\text{maximize}}
& &  \sum_k w_k\log_2\left(1+\frac{Mp_d^k \gamma_k }{1+\sum_{j= 1}^K \beta_j p_d^j}\right)   \\
& \text{subject to}
& & \tau_p p_{p}^k+(T-\tau_p) p_d^k \leq E_{k}, \forall k, \\
& & & p_d^k\geq0,~p_p^k\geq0, \forall k.
\end{aligned}
\end{equation}
Since $\gamma_k$ depends on $p_p^k$ which is also an optimization variable, the problem is non-convex. However we find out that the tools we developed for the max-min problem help us here as well. More specifically, we make use of Proposition \ref{pusolution} with $w_k=1~\forall k$. Define $x_k=p_d^k\gamma_k$ as the SP of user $k$, then we use Proposition \ref{pusolution} to make a change of variables in \eqref{sumratejoint} and use the same techniques as in the case of optimizing data power only. We obtain the following theorem:

\begin{theorem}\label{summrc}Problem \eqref{sumratejoint} can be reformulated into the following form:
\begin{equation}\label{sumratejoint2}
\begin{aligned}
& \underset{~s,~\{y_k\}}{\mathrm{maximize}}
& &  \sum_k w_k\log_2\left(1+M y_k \right)   \\
& \mathrm{subject~to}
& & \sum_{j= 1}^K \beta_jq(y_j,s) \leq 1-s,
\end{aligned}
\end{equation}
where $q(y_j,s)$ is defined in \eqref{qdef} on top of next page. The two formulations are equivalent in the sense that they have the same optimal objective values, and the solution to \eqref{sumratejoint} can be obtained from solution to \eqref{sumratejoint2} via $p_d^k=q(y_k,s)/s$. Moreover problem \eqref{sumratejoint2} is jointly convex in $s$ and $y_k$.
\end{theorem}
\begin{figure*}[!t]
\begin{equation}\label{qdef}
q(y_j,s)=\frac{E_{j}\beta_js +(T-K)y_j-\sqrt{E_{j}^2\beta_j^2s^2-2(T-K)(E_{j}\beta_j+2)y_js+(T-K)^2y_j^2}}{2(T-K)\beta_j}.
\end{equation}
\hrulefill
\end{figure*}

\begin{proof} First we introduce a dummy variable $t$ and rewrite \eqref{sumratejoint} as
\begin{equation}
\begin{aligned}
& \underset{~t,~\{p_d^k\},~\{p_p^k\}}{\text{maximize}}
& &  \sum_k w_k\log_2\left(1+\frac{Mp_d^k \gamma_k }{t}\right)   \\
& ~\text{subject to}
& & \tau_p p_{p}^k+(T-\tau_p) p_d^k \leq E_{k}, \forall k, \\
& & & p_d^k\geq0,~p_p^k\geq0 ~\forall k \\
& & & 1+\sum_{j= 1}^K \beta_j p_d^j \leq t.
\end{aligned}
\end{equation}
The last constraint is relaxed from equality to inequality without changing the solutions to the problem. This is because the objective function is monotonically decreasing in $t$, thus at the optimal point the last inequality will always be active.
Next we apply Proposition \ref{pusolution} with $w_k=1~\forall k$. Define $x_k=p_d^k\gamma_k$ as the SP of user $k$ to obtain the following problem
\begin{equation}
\begin{aligned}
& \underset{~t,~\{x_k\}}{\text{maximize}}
& &  \sum_k w_k\log_2\left(1+\frac{Mx_k }{t}\right)   \\
& \text{subject to}
& & 1+\sum_{j= 1}^K \beta_j r(x_j) \leq t,
\end{aligned}
\end{equation}
where $r(x_j)$ is defined in \eqref{rdef} on top of next page.

Finally we apply the variable substitution $y_k=x_k/t$ and $s=1/t$ to obtain \eqref{sumratejoint2}.

From the proof of Theorem \ref{snr} we can deduce that $r(x_j)$ is a convex function in $x_j$. Next we observe that $q(y_j,s)$ is a perspective transformation of $r(x_j)$ and therefore preserve the convexity\cite{BV}. Hence we conclude that \eqref{sumratejoint2} is jointly convex in $y_j$s and $s$.
\end{proof}

\begin{figure*}[!t]
\begin{equation}\label{rdef}
r(x_j)=\frac{E_{j}\beta_j +(T-K)x_j-\sqrt{E_{j}^2\beta_j^2-2(T-K)(E_{j}\beta_j+2)x_j+(T-K)^2x_j^2}}{2(T-K)\beta_j}.
\end{equation}
\hrulefill
\end{figure*}
Since we have the convex reformulation \eqref{sumratejoint2} we can use standard convex solvers to find the optimal solutions efficiently, and the optimal power control parameters can be recovered easily. Here we use the MOSEK solver \cite{MOSEK} with CVX \cite{CVX}

\subsection{Sum SE for ZF}
In the case of perfect CSI, maximizing sum SE for ZF is straightforward. This is because the ZF detector completely removes all the interference from other users and creates $K$ parallel channels. However in the case of imperfect CSI, the interference is reduced but still remains, which makes the sum SE problem at least as difficult as with MRC.  Fortunately, the techniques we developed for solving the MRC case can be applied here to solve the problem to global optimal. Similarly we will first describe the case of optimizing data power only and then extended to joint pilot and data power optimization.

By using Theorem \ref{theorem1}, \eqref{opt} now becomes the following optimization problem:
\begin{equation}\label{sumratezf}
\begin{aligned}
& \underset{\{p_p^k\},~\{p_d^k\}}{\text{maximize}}
& &  \sum_k w_k\log_2\left(1+\frac{(M-K)p_d^k \gamma_k}{1+\sum_{j=1}^K  p_d^j(\beta_j-\gamma_j)}\right)   \\
& \text{subject to}
& & \tau_p p_{p}^k+(T-\tau_p) p_d^k \leq E_{k}, \forall k \\
& & & p_p^k\geq 0, p_d^k\geq0, \forall k.
\end{aligned}
\end{equation}

\subsubsection{Optimizing Data Power}
In the ZF case, we have the following the following problem:
\begin{equation}\label{sumratedatazf}
\begin{aligned}
& \underset{\{p_d^k\}}{\text{maximize}}
& &  \sum_k w_k\log_2\left(1+\frac{(M-K)p_d^k \gamma_k }{1+\sum_{j= 1}^K (\beta_j-\gamma_j) p_d^j}\right)   \\
& \text{subject to}
& & p_d^k \leq P_{k}, \forall k \\
& & & p_d^k\geq0, \forall k. \\
\end{aligned}
\end{equation}
We observe that this problem has exactly the same structure as \eqref{sumratedata2} in the MRC case where only the constant $\beta_j$ changes to $\beta_j-\gamma_j$. Therefore same analysis and algorithm applies here where we substitute all $\beta_j$ with $\beta_j-\gamma_j$.

\subsubsection{Joint Pilot and Data Power Optimization}
Next we extend this result to the case of joint power control over pilot and data power. The problem is as follows:
\begin{equation}\label{sumratejointzf}
\begin{aligned}
& \underset{\{p_d^k\},~\{p_p^k\}}{\text{maximize}}
& &  \sum_k w_k\log_2\left(1+\frac{(M-K)p_d^k \gamma_k }{1+\sum_{j= 1}^K (\beta_j-\gamma_j) p_d^j}\right)   \\
& \text{subject to}
& & \tau_p p_{p}^k+(T-\tau_p) p_d^k \leq E_{k}, \forall k \\
& & & p_d^k\geq0,~p_p^k\geq0, \forall k.\\
\end{aligned}
\end{equation}
The transformation we did in the MRC case can be applied here as well as proved by the following corollary:
\begin{corollary}\label{sumzf} Problem \eqref{sumratejointzf} can be reformulated into the following form:
\begin{equation}\label{sumratejointzf2}
\begin{aligned}
& \underset{~s,~\{y_k\}}{\mathrm{maximize}}
& &  \sum_k w_k\log_2\left(1+(M-K) y_k \right)   \\
& \mathrm{subject~to}
& & \sum_{j= 1}^K \beta_jq(y_j,s)-\sum_{j=1}^K y_j \leq 1-s,
\end{aligned}
\end{equation}
where $q(y_j,s)$ is given in \eqref{qdef} which is the same as in the MRC case. The two formulations are equivalent in the sense that they have the same optimal objective values, and the solution to \eqref{sumratejointzf} can be obtained from solution to \eqref{sumratejointzf2} via $p_d^k=q(y_k,s)/s$. Moreover problem \eqref{sumratejointzf2} is jointly convex in $s$ and $y_k$.
\end{corollary}

\begin{proof}
The only difference compared with the case of MRC is that $\beta_j$ changes to $\beta_j-\gamma_j$ in all expressions. The proof is similar to the case of MRC, and is omitted here for brevity.
\end{proof}

\section{Extension to Correlated Fading Channels}\label{correlated}

In this section, we extend our results to case of correlated fading
channels. We only consider weighted max-min fairness for MRC here, to
exemplify how our techniques in the previous sections apply to other
channel models. The other cases are left for future work.

For the correlated fading channels, we model $g_k\sim CN(0, \bR_k)$
where the covariance matrix $\bR_k$ characterizes the spatial
correlation. The large-scale fading is the same for all antennas so
all diagonal entries are equal to $\beta_k$.  The MMSE channel estimation
requires the storage of the entire matrix $\bR_k$ for every user,  and the estimation
requires the inversion of large matrices -- which has a high associated complexity. To
avoid this complexity, we adopt the element-wise MMSE estimator
proposed in \cite{BG2006}. During the training phase, the BS receives
the pilot signals, correlates them with pilot sequence of user $k$ and
obtains
\begin{equation}
\by_k=\sqrt{\tau_p p_p^k}\bg_k + \bn_p, \quad k=1,\ldots,K.
\end{equation}
The estimate is then
\begin{equation}
\hat{\bg}_k=\frac{\sqrt{\tau_p p_p^k}\beta_k}{1+\tau_p p_p^k\beta_k} \by_k,\quad k=1,\ldots,K.
\end{equation}
This estimate, $\hat{\bg}_k$, is for linear detection of data from
user $k$. With this channel model and estimation method, we obtain the
following achievable SE:

\begin{lemma}\label{correlatedse}
The capacity of user $k$ with MRC detection under correlated fading and
element-wise MMSE estimation is lower bounded by the achievable ergodic SE
\begin{equation}
R_k^{corr}=\left(1-\frac{\tau_p}{T}\right)\log_2(1+\mathrm{SINR}_k^{corr})
\end{equation}
where pilot and payload powers are arbitrary,
\begin{equation}\label{sinrmrccor}
\mathrm{SINR}_k^{corr}=\frac{Mp_d^k \gamma_k }{1+\sum_{j= 1}^K \mathrm{tr}(\bR_j\bR_k) p_d^j \frac{\gamma_k}{M\beta_k^2}+\sum_{j=1}^K p_d^j \frac{\beta_j}{1+\tau_pp_p^k\beta_k}}
\end{equation}
and $\gamma_k=\frac{\tau_p p_p^k \beta_k^2}{1+\tau_p p_p^k \beta_k}$.
\end{lemma}

\begin{proof}
  The proof is given in Appendix \ref{proofcorr}.
\end{proof}

We observe that Theorem \ref{theorem1} for the optimal training length
can be easily extended to cover this case, and therefore the
optimization problem we are interested to solve is:
\begin{equation}\label{maxmincor}
\begin{aligned}
& \underset{\{p_p^k\},~\{p_u^k\}}{\text{maximize}}
& & \underset{k}{\text{min}} ~w_k\mathrm{SINR}_k^{corr}   \\
& \text{subject to}
& & \tau_p p_{p}^k+(T-\tau_p) p_u^k \leq E_{k}, \forall k \\
& & & p_p^k\geq 0, p_u^k\geq0, \forall k.
\end{aligned}
\end{equation}
The epigraph form of \eqref{maxmincor} is
\begin{equation}\label{maxmincor2}
\begin{aligned}
& \underset{\{p_p^k\},\{p_d^k\},~\lambda}{\text{maximize}}
& &  \lambda  \\
& ~\text{subject to}
& & w_kMp_d^k p_p^k \beta_k^2 \tau_p \geq\\
&&& \lambda(1+\tau_p \beta_k p_p^k+ \sum_{j=1}^K \beta_j p_d^j+ \\
&&& \tau_p p_p^k\sum_{j= 1}^K \mathrm{tr}(\bR_j\bR_k) p_d^j \frac{1}{M}), \forall k\\
& & & \tau_p p_{p}^k+(T-\tau_p) p_d^k \leq E_{k}, \forall k \\
& & & p_p^k\geq 0, p_d^k\geq0, \forall k.
\end{aligned}
\end{equation}
We recognize \eqref{maxmincor2} as a GP and therefore it can be solved
efficiently, using general purpose solvers.

\color{black}
\section{Simulation Results and Discussion}\label{simulation}
In this section we present simulation results to demonstrate the benefits of our algorithms and compare the performance with the case of no power control (i.e.,  full equal power) as well as the case of power control on the payload power only (and full power pilots).  We consider a scenario with $M=100$ antennas, $K_0=10$ users, and the length of the coherence interval is $T=200$ (which for example corresponds to a coherence bandwidth of $200$ kHz and a coherence time of $1$ ms). The users are assumed to be uniformly and randomly distributed in a cell with radius $R=1000$ m and no user is closer to the BS than $100$ m. The path-loss model is chosen as $\beta_k=z_k/r_k^{3.76}$ where $r_k$ is the distance of user $k$ from the BS where $z_k$ represents the independent shadowing effect. Shadowing is chosen to be log-normal distributed with a standard deviation of $8$ dB. Due to the long tail behavior of the log-normal distribution there could be some users with very small $\beta_k$, therefore in each snapshot the user with the smallest $\beta_k$ is dropped from service. Therefore the algorithm is run for $K=K_0-1=9$ users.

The energy budgets $E_{k}=10^{-0.5}\times R^{3.76} \times T$ and $E_{k}=10^{0.5}\times R^{3.76} \times T$ give a median SNR of $-5$ dB and $5$ dB at the cell edge when using equal power allocation. The weights $w_k$ are set to be equal in all the simulations. The algorithms are run for $1000$ Monte-Carlo simulations where in each snapshot the users are dropped randomly in the cell so that the large-scale fading $\beta_k$ changes.

\subsection{Max-Min SE Results}
We compare 4 schemes: 1) the solution to problem \eqref{maxmin} (marked as `Max-min' in the figures); 2) equal power allocation $p_d^k=p_p^k=E_{k}/T$ (marked as `Equal Power' in the figures); 3) optimizing only payload power for problem \eqref{maxmin} by fixing $p_p^k=E_{k}/T$ (marked as `Max-min (data)' in the figures); 4) the scheme that maximizes the sum SE is presented as well for reference (marked as `sum' in the figures). The same schemes are tested for both MRC and ZF, and low and high SNR scenarios.

\begin{figure*}[!t]
\begin{center}
\includegraphics[width=\textwidth]{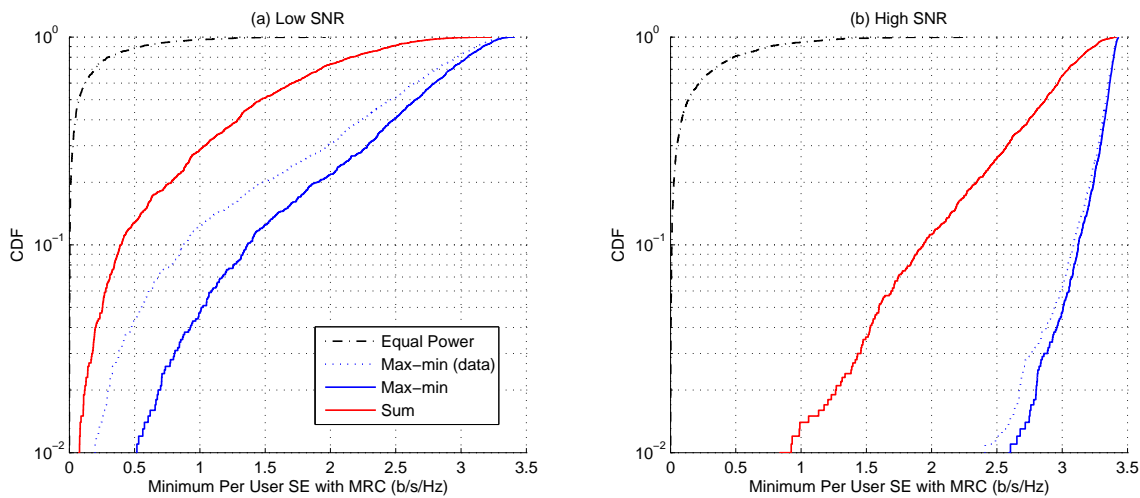}
\caption {\label{minratemrc} CDF of the minimum SE with $M=100$, $K_0=10$, $T=200$, $R=1000$ m for MRC. Subplots (a) and (b) correspond to low SNR ($-5$ dB) and high SNR ($5$ dB) at the cell edge, respectively.}
\end{center}
\end{figure*}

\begin{figure*}[t]
\begin{center}
\includegraphics[width=\textwidth]{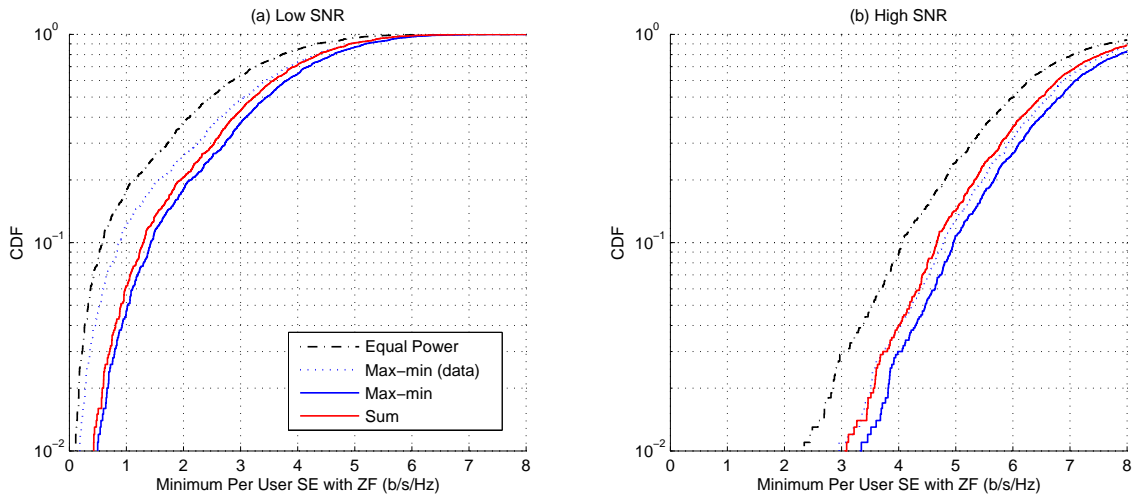}
\caption {\label{minratezf} CDF of the minimum SE with $M=100$, $K_0=10$, $T=200$, $R=1000$ m for ZF. Subplots (a) and (b) correspond to low SNR ($-5$ dB) and high SNR (5 dB) at the cell edge, respectively.}
\end{center}
\end{figure*}

In Figure \ref{minratemrc} (a) and (b) we plot the cumulative distribution function (CDF) of the minimum SE over different snapshots of user locations for MRC at low and high SNR respectively. We observe that without any power control in almost all of the cases the user with the lowest SNR will get less than $0.5$ bit/s/Hz in both low and high SNR scenarios. This is not acceptable if we want to provide decent quality of service to every user being served. With max-min power control for both pilot and data we resolve this problem by guaranteeing the users an SE of more than $1$ bit/s/Hz with $0.95$ probability and $2.75$ bit/s/Hz with $0.5$ probability. In low SNR scenarios the joint optimization doubles the $0.95$ likely point, from $0.5$ to $1$ bit/s/Hz, which proves the need of joint pilot and data power optimization at low SNR. In this case with data power control the user with the worst channel would have poor channel estimates that limits the SE, while with joint power control they borrow power from the data part to enhance channel estimation and thereby increase the SE. However in the high SNR scenarios the gain is marginal by the joint optimization, power control over data is enough. This is because the channel estimates are already good enough for linear detection. The performance of the sum SE formulation is not surprising as it is not designed for improving the minimum SE. It boosts the SE of the users with better channels to increase the sum SE, which in turn scarifies the users with worse channels.

In Figure \ref{minratezf} (a) and (b) we plot the CDF of the minimum SE over different snapshots of user locations for ZF at low and high SNR respectively. We observe that all schemes perform similarly and the gains from joint power control with respect to only power control over data are not as large as in the case of MRC. This is because with ZF most interference is removed by the detector, however in low SNR scenarios joint power control is still necessary as it increases the $0.95$ likely point from $0.5$ to $1$ bit/s/Hz compared to power control over data only. The performance of the sum SE formulation is surprisingly good at both low and high SNR and is even better than the max-min scheme with only data power control. This suggests that with ZF detector we can go for the sum SE formulation and push up the total system throughput without sacrificing much of the worse users' performance.

\subsection{Sum SE Results}
We compare $4$ schemes: 1) the scheme that maximizes the sum SE (marked as `Sum' in the figures); 2) equal power allocation $p_d^k=p_p^k=E_{k}/T$ (marked as `Equal Power' in the figures); 3) optimizing the data power only for sum SE by fixing $p_p^k=E_{k}/T$ (marked as `Sum (data)' in the figures); 4) the max-min scheme is also presented for reference (marked as `max-min' in the figures). The same schemes are tested for both MRC and ZF.
\begin{figure*}[!t]
\begin{center}
\includegraphics[width=\textwidth]{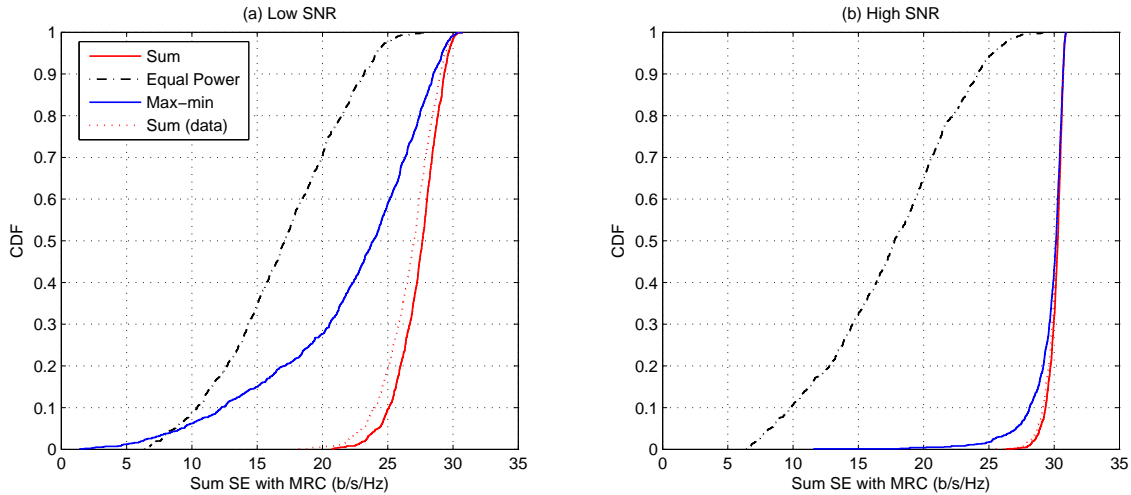}
\caption {\label{sumratemrcplot}CDF of the sum SE with $M=100$, $K_0=10$, $T=200$, $R=1000$ m for MRC. Subplots (a) and (b) correspond to low SNR ($-5$ dB) and high SNR ($5$ dB) at the cell edge, respectively.}
\end{center}
\end{figure*}

\begin{figure*}[!t]
\begin{center}
\includegraphics[width=\textwidth]{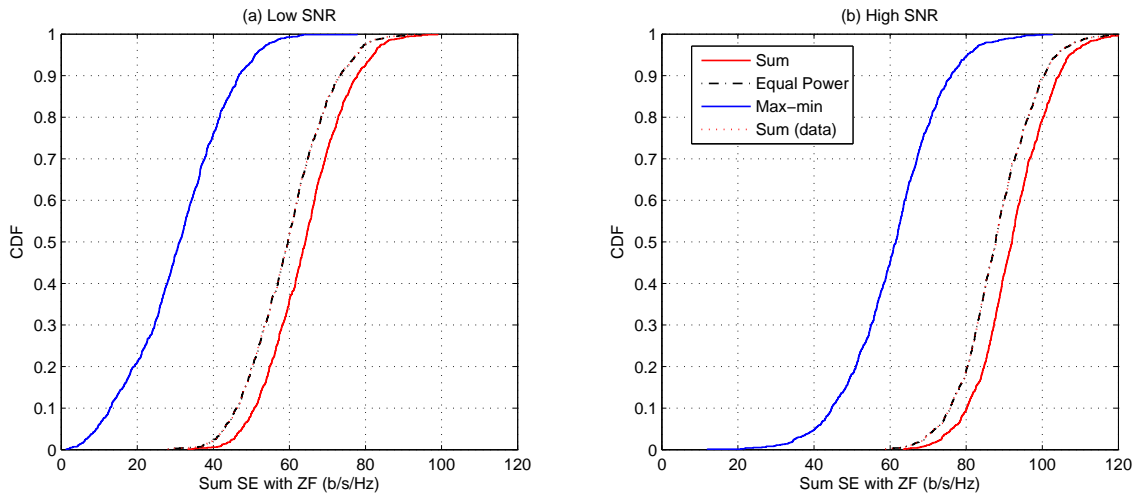}
\caption {\label{sumratezfplot}CDF of the sum SE with $M=100$, $K_0=10$, $T=200$, $R=1000$ m for ZF. Subplots (a) and (b) correspond to low SNR ($-5$ dB) and high SNR ($5$ dB) at the cell edge, respectively.}
\end{center}
\end{figure*}

In Figure \ref{sumratemrcplot} (a) and (b) we plot the CDF of the sum SE for the scenario we described above for MRC at low and high SNR respectively. We observe the optimized power control increases the sum SE significantly. The whole CDF is shifted to the right by almost $15$  bit/s/Hz in the low SNR scenario with the proposed power control as compared to equal power allocation. At low SNR the joint power control offers about $10\%$ increase over the case with only data power control. At high SNR the gain is marginal as the SEs of the users have saturated so we are in the $\log$ part of the SE already. The max-min scheme performs well at high SNR due to the saturation of SE, but worse at low SNR. This is because enforcing max-min fairness lead to large sacrifices in sum SE at low SNR. The reason is that with high probability there will be some very disadvantaged user, and everyone else has to cut back significantly to avoid causing near-far interference.

In Figure \ref{sumratezfplot} (a) and (b) we plot the CDF of the sum SE for ZF at low and high SNR respectively. We observe that with ZF when we optimize only the data power the optimal scheme is always using full power. The reason for this is that in single cell systems ZF removes most of the interference, the near-far effects are almost removed by the ZF detector thus creating almost parallel channels. Therefore the scheme with equal power allocation is the same as optimizing data power only. The joint power control offers about $10\%$ improvements over the case with only data power control at low SNR and the gain diminish as the SNR increases. However there will always be a gap between the two schemes, this is because even when the SNR tends to infinity we can use always save power on the pilot and use it for data which increases the SE. The max-min scheme performs poorly in both scenarios, this confirms our suggestion that with ZF we should use the sum SE formulation.

\subsection{Robustness}\label{robustness}

In this subsection, we present simulation results for the case when
the large scale fading parameters are not known perfectly, but
obtained through estimation. We assume that the BS collects $N$
processed pilots from each user to perform this estimation.
Specifically, denoting each channel realization by $\bg_k^i$, the
processed pilot signals received by the BS for each user can be
written as
\begin{equation}
\by_k^i=\sqrt{\tau_p p_k}\bg_k^i+\bw_k^i, i=1,\ldots,N,
\end{equation}
where $\by_k^i$ is the processed received signal, $\tau_p$ is the length
of the pilot, $p_k$ is the signal power and $\bw_k^i$ is additive
noise with variance $1$.
Then we estimate $\beta_k$ as follows:
\begin{equation}
\hat{\beta}_k=\frac{\sum_{i=1}^N ||\by_k^i||^2-MN}{MN\tau_p p_k}.
\end{equation}
This estimate is justified by the fact that
\begin{equation}
\begin{aligned}
||\by_k^i||^2 & \approx \tau_p p_k ||\bg_k^i||^2+||\bw_k^i||^2 \\
& \approx\tau_p p_k \beta_k M +M.
\end{aligned}
\end{equation}

Figure \ref{estimatedbeta} shows the minimum SE achieved by our
max-min scheme with the proposed estimator of the large-scale fading
parameters. The number of observations is $N=10$ and the median SNR at
the cell edge ranges from $-10$ dB to $10$ dB; all other simulation
parameters are the same as in the previous subsection. The estimated
$\beta$s are treated as the true $\beta$s in the optimization (marked
as 'Estimated'). The performance is compared with the case when the
$\beta$s are known perfectly (marked as 'Genie Aided'). We observe
that with the simple, above suboptimal estimator and the small number
of training symbols, the performance degradation is almost
negligible. We conclude that our scheme shows significant robustness
against estimation errors in the large-scale fading parameters.

\begin{figure}
\begin{center}
\includegraphics[width=0.5\textwidth]{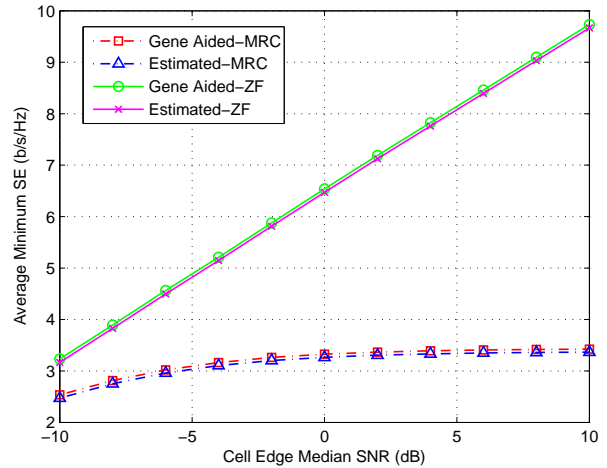}
\caption {\label{estimatedbeta} Average minimum SE with $M=100$, $K_0=10$, $T=200$, $N=10$, $R=1000$ m for estimated large scale fading parameters.}
\end{center}
\end{figure}

\subsection{Correlated Fading}

In this subsection we look at the performance of joint pilot and data
power control in correlated fading channels. We use the one-ring model
\cite{ANAG2013} to model the correlations. An angular spread of $10$
degrees is chosen, and the angles of arrival of different users are
independent and uniformly distributed between $0$ and $180$
degrees. The median cell edge SNR is $-10$ dB and all other parameters
are the same as in the above.

Figure \ref{correlatedsim} shows the CDF of the minimum SE achieved by
our scheme with element-wise MMSE channel estimation and MRC.  We
compare $4$ schemes: 1) the solution to problem \eqref{maxmincor2}
(marked as `Max-min'); 2) equal power allocation $p_d^k=p_p^k=E_{k}/T$
(marked as `Equal Power'); 3) optimization of only payload power for
problem \eqref{maxmincor2}, by fixing $p_p^k=E_{k}/T$ (marked as
`Max-min (data)'); 4) the solution to problem \eqref{maxmin} but with
application of the power control parameters obtained under the
i.i.d. assumption (marked as `Max-min i.i.d.'). From the plot we see
similar behaviors as in the i.i.d. channels, that is, joint pilot and
data power control improves the minimum SE substantially.  Directly
applying the power control parameters obtained under the
i.i.d. assumption, neglecting the correlation, yields surprisingly
good performance.

Taken together, joint pilot and data power control is highly useful in
the low SNR regime also for correlated fading channels. We expect this
conclusion to hold also for other channel models, which have to be
left for future work.

\begin{figure}
\begin{center}
\includegraphics[width=0.5\textwidth]{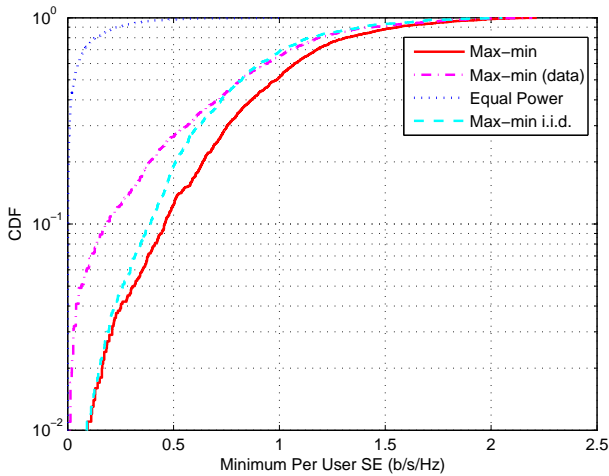}
\caption {\label{correlatedsim} CDF of the minimum SE with $M=100$, $K_0=10$, $T=200$, $R=1000$ m for MRC in correlated fading. The median SNR is $-10$ dB at the cell edge.}
\end{center}
\end{figure}

\subsection{Dependence on SNR, K and T}

In this subsection we study the dependence of the gain of joint
pilot and data power control on the SNR, number of users and length of
coherence interval.

First we investigate the performance of MRC. At low SNR, the noise
dominates over interference and we can approximate \eqref{sinrmrc} as
\begin{equation}
\mathrm{SINR}_k \approx Mp_d^k\gamma_k \approx MK p_p^kp_d^k\beta_k^2,
\end{equation}
for $\tau_p=K$.
Under the power constraint
\begin{equation}
Kp_p^k+(T-K)p_d^k \leq E_k,
\end{equation}
it is straightforward to show that
\begin{equation}
p_p^k=\frac{E_k}{2K} ~\text{and}~ p_d^k=\frac{E_k}{2(T-K)}
\end{equation}
maximize the approximate SINR. Since $T\gg K$, this means that the user
allocates substantially more power to pilots than to data at low
SNR. Compared to the case of data power control only, where
$p_p^k=p_d^k=E_k/T$, we have
\begin{equation}
\mathrm{SINR}_k^{opt}\approx \frac{T^2}{4K(T-K)} \mathrm{SINR}_k^{data},
\end{equation}
where $\mathrm{SINR}_k^{opt}$ represents the SINR  obtained by
optimizing the pilot and data power and $\mathrm{SINR}_k^{data}$
represents the SINR obtained by only optimizing the data power.

To conclude, the gain of joint pilot and data power control can be
substantial at low SNR, and when $K$ is small relative to $T$. For ZF,
similar results are obtained at low SNR, where the interference can be
neglected. Therefore our scheme may be particularly useful for
wireless broadband access with stationary terminals, as in that
application the coherence interval is usually very long.

At high SNR, when interference dominates over noise,
$\gamma_k\approx\beta_k$. Then the impact of the pilot power is
negligible for MRC. However for ZF, the interference is cancelled out
completely, thus creating parallel channels. More power will be spent
on data to boost the SE. However the gain will not be substantial as
the spectral efficiency only grows logarithmically with SINR in this
regime. Therefore optimizing the data power is most important.

\subsection{Complexity}

In this subsection, we characterize the computational complexity of
our schemes, and compare it to that of the other digital signal
processing that is carried out in massive MIMO systems (in particular,
channel estimation and linear detection of the data). We perform the
comparison for MRC, as ZF would consume more computational resources.

Since our power control parameters are computed based on the large
scale fading, we only have to recompute them at the pace that the
large-scale fading changes. The complexity of our algorithm for the
max-min problem is of order $O(K)$. The sum SE problem is transformed
to a convex problem that can be solved by a general interior point
method. Its complexity is $N_{it}O((m+K)^3)$, where $N_{it}$ is the
number of Newton iterations required to achieve a predetermined
precision, and $m$ is the number of constraints in the problem. The
exact number of $N_{it}$ is hard to determine, however in practice,
$N_{it}$ is typically in the order of tens \cite[Chapter
  11]{BV}. Therefore $100$ should be a good enough bound; in any case,
the algorithm may be terminated after $100$ iterations. In each Newton
iteration we are solving a linear system of equations. Since we have
$2K+1$ constraints and $K+1$ variables, the number of operations
required for solving this Newton system is about $9K^3$, assuming the
use of Cholesky factorization.  In the channel estimation phase the
number of operations is approximately $2MK^2$, and for MRC detection
the number of operations is approximately $2MK$ per data symbol
\cite{BSHD2015}. Therefore the total amount of computations in one
coherence interval is approximately $2MK^2+2MK(T-K)=2MKT$. The
measurements reported in \cite{VHU} show that the large-scale fading
parameters are constant over a duration that is on the order of 100
times the channel coherence time. Moreover, for the sake of argument,
we assume that there are 100 subcarriers in the system.  These
assumptions result in a relative computational overhead of the
proposed sum SE algorithm as
$\frac{N_{it}9K^3}{20000MKT+N_{it}9K^3}$. We see that even with
$N_{it}=100$ (likely an overestimate) this overhead is on the order of
$0.02\%$.

We conclude that while the complexity calculation given here
represents a first-order estimate only, the extra efforts for solving
the joint optimization problem is negligible in representative cases.

\color{black}\section{Conclusion}\label{conclusion}

We considered the optimal joint pilot and data power allocation
problems in single cell uplink massive MIMO systems with MRC or ZF
detection. It was first proved that the optimal length of the training
interval equals the number of users. Using the SE as performance
metric and setting a total energy budget, the power control was
formulated as optimization problems for two different objective
functions: the weighted minimum SE and the weighted sum SE. The
optimal power control policy was found for the case of maximizing the
weighted minimum SE by a semi-closed form solution to a single
variable equation with unique solution. The optimal power control
parameters were shown to be the same for MRC and ZF. For maximizing
the sum SE a convex reformulation was found and efficient solution
algorithms were developed. The methods have also been
  extended to handle the case of correlated fading, although a
  complete treatment of all aspects of that case is left for future
  work.

Simulation results demonstrated the advantage of joint optimization
over both pilot and data power, and how the two objectives behave at
low and high cell-edge SNRs. With MRC we have a clear choice to make
between max-min and sum SE, which is dependent on the system
requirements. With ZF we can maximize sum SE without sacrificing much
in min SE. The need of joint pilot and data power control is
particularly important at low SNR, while at high SNR optimizing only
data power seems to be good enough. Since multi-cell systems are
interference-limited, we predict that we will get results similar to
the low SNR results, particularly if a large pilot reuse factor is
used to get single-cell-like estimation quality. The
  numerical results were also justified by a theoretical analysis in
  the low and high SNR regime. This analysis showed that the gain is
  more substantial when the number of users, $K$, is small compared to
  the length of the coherence interval, $T$.

Future work includes extension of the methodologies to multi-cell
systems and more sophisticated system models.

\begin{appendices}
\subsection{Proof of Theorem \ref{theorem1}}\label{proofpilot}
Before proving Theorem \ref{theorem1} we state and prove two lemmas.
\begin{lemma}\label{lemma2}
For any $x\geq 0$, we have $\ln (x) \geq \frac{x-1}{x}$ with equality if and only if $x=1$.
\end{lemma}
\begin{proof}
Write $f(x)=\ln (x) - \frac{x-1}{x}$, then we have $f'(x)=\frac{1}{x}-\frac{1}{x^2}$. Observing that $f'(x)\leq 0, \forall x \in(0,1]$ and $f'(x)\geq0, \forall x>1$, we can conclude that $x=1$ is the minimum point of $f(x)$ at which $f(x)=0$. Thus we have $f(x)\geq 0, \forall x>0$, which proves the lemma.
\end{proof}
\begin{lemma}\label{lemma3}
For any positive constants $a,~b$ and $c$, $g(x)=x\log_2\left(1+\frac{a}{bx+c}\right)$ is a strictly monotonic increasing function in $x$ for all $x>0$.
\end{lemma}
\begin{proof}
Taking the first derivative we have
\begin{equation}
\begin{aligned}
g'(x)&=\frac{1}{\ln(2)}\ln\left(1+\frac{a}{bx+c}\right)\\
&+\frac{x}{\ln(2)}\cdot\frac{1}{1+a/(bx+c)}\cdot\frac{-a}{(bx+c)^2\cdot b}\\
&=\frac{1}{\ln(2)}\left(\ln\left(1+\frac{a}{bx+c}\right)-\frac{abx}{(bx+c)(bx+c+a)}\right)\\
&> \frac{1}{\ln(2)}\left(\ln\left(1+\frac{a}{bx+c}\right)-\frac{a}{bx+c+a}\right)\geq 0.
\end{aligned}
\end{equation}
The first inequality comes from the fact that $bx/(bx+c)<1$ for any strictly positive $b$ and $c$. The last inequality follows from putting $1+a/(bx+c)$ in Lemma \ref{lemma2}.
\end{proof}
Next, we prove Theorem \ref{theorem1} by contradiction. Assume that $\tau_p^*>K$, $p_p^{k*}$ and $p_d^{k*}$  is the optimal solution to problem \eqref{opt}. From Lemma \ref{lemma1} we know that
\begin{equation}
\tau_p^* p_{p}^{k*}+(T-\tau_p^*) p_d^{k*} = E_{k}, ~k=1,\ldots,K.
\end{equation} We will now construct a new feasible point that gives a higher objective function. Choose $\tau_p^{'}=K$, $p_p^{k'}=\tau_p^* p_p^{k*}/K$, $p_d^{k'}=(E_{k}-\tau_p^* p_p^{k*})/(T-K)$ for every user $k$, then $\gamma_k^{'}=\gamma_k^{*}$ as $\tau_p^{'}p_p^{k'}=\tau_p^* p_p^{k*}$. We compare the value of $R_k(\tau_p, p_p^k, p_d^k)$ for these two sets of parameters. The achievable SE for user $k$ with our new construction is
\begin{equation}
R_k(K,p_p^{k'},p_d^{k'})=\left(1-\frac{K}{T}\right)\log_2\left(1+\frac{a_k}{T-K+c_k}\right)
\end{equation}
where $a_k=M\gamma_k^{'}(E_{k}-\tau_p^* p_p^{k*})$, $c_k=\sum_{j=1}^K\beta_j(E_{j}-\tau_p^* p_p^{j*})$ for the MRC, and $a_k=(M-K)\gamma_k^{'}(E_{k}-\tau_p^* p_p^{k*})$, $c_k=\sum_{j=1}^K(\beta_j-\gamma_j^{'})(E_{j}-\tau_p^* p_p^{j*})$ for the ZF.
Then we observe that
\begin{equation}
R_k(\tau_p^*,p_p^{k*},p_d^{k*})=\left(1-\frac{\tau_p^*}{T}\right)\log_2\left(1+\frac{a_k}{T-\tau_p^*+c_k}\right).
\end{equation}
Next we apply Lemma \ref{lemma3} with $x=T-\tau_p$ we can know that $TR_k$ is a strictly monotonic increasing function in $T-\tau_p$. Therefore $R_k(K,p_p^{k'},p_d^{k'})>R_k(\tau_p^*,p_p^{k*},p_d^{k*})$ and the maximum is achieved at $\tau_p=K$ due to the constraint $\tau_p \geq K$. This is a contradiction to the assumption that $\tau_p^*>K$, hence $\tau_p^*=K$. Since this holds for every $k$, we have proved the theorem for any monotonic increasing function $U(R_1,\ldots,R_k)$.
\subsection{Proof of Theorem \ref{v_wf}}\label{proofvwf}
We first state the Lagrangian function of problem \eqref{sumratedata2}:
\begin{equation}
\begin{aligned}
&L(s,\{x_k\},\{\lambda_k\},\{\mu_k\},\nu)=\sum_k w_k \log_2(1+a_kx_k)\\
&-\sum_k\lambda_k(x_k-\beta_kP_ks)+\sum_k\mu_kx_k-\nu(\sum_kx_k+s-1).
\end{aligned}
\end{equation}
Then we can write the KKT conditions\cite{BV} for problem \eqref{sumratedata2}:
\begin{eqnarray}
\begin{aligned}
\frac{w_k}{\frac{1}{a_k}+x_k}-\lambda_k+\mu_k-\nu&=0, \forall k, \\
\lambda_k\geq0,~x_k \leq \beta_kP_ks, ~\lambda_k(x_k-\beta_kP_ks)&=0, \forall k, \\
\mu_k\geq0,~x_k \geq 0,~ \mu_kx_k&=0, \forall k, \\
\sum_kx_k&=1-s, \label{sumpower}\\
\sum_k\lambda_k\beta_kP_k-\nu&=0.\label{conditions}
\end{aligned}
\end{eqnarray}
We construct a set of solution to the above KKT conditions as follows:
\begin{align}
x_k&=\min \left(\beta_kP_ks,\left(\frac{w_k}{\nu}-\frac{1}{a_k}\right)^{+}\right), \forall k, \label{1}\\
\lambda_k&=\left(\frac{w_k}{\frac{1}{a_k}+x_k}-\nu\right)^{+}, \forall k, \label{2}\\
\mu_k&=(\nu-a_k)^{+}, \forall k. \label{3}
\end{align}
We can easily verify this set of solutions together with condition \eqref{sumpower} and \eqref{conditions} satisfies the overall KKT conditions. When $s$ is considered to be a constant, the last condition of \eqref{conditions} is not necessary as it corresponds to $\frac{\partial L}{\partial s}=0$. This set of solutions is a function of $\nu$ and we are looking for $\nu$ such that \eqref{sumpower} is satisfied. For a given $s$, finding the optimal $x_k$s and $\nu$ can be done using algorithms in \cite{PF2005} and \cite{HZZN2013}. Then we perform bisection on $s$ to find the optimal $s$ that satisfies \eqref{conditions}. Using bisection we are looking for the zero crossing point of a univariate function, and this requires the function to have different signs on each end of the interval. To justify that we can use bisection, we need to check the sign of $f(s)\triangleq \sum_{j=1}^K\beta_jP_j\left(\frac{w_j}{\frac{1}{a_j}+x_j}-\nu\right)^{+}-\nu$ on the boundaries, which corresponds to checking $s=\frac{1}{1+\sum_j\beta_jP_j}$ and $s=1$.

When $s=\frac{1}{1+\sum_j\beta_jP_j}$, then to satisfy \eqref{sumpower} $x_k=\beta_kP_ks$, and thus $\lambda_k\geq0, ~\forall k$ and $\mu_k=0, ~\forall k$. This is equivalent to
\begin{align}
\nu &\leq \frac{w_k}{\frac{1}{a_k}+\frac{\beta_kP_k}{1+\sum_j \beta_j P_j}}, ~\forall k.\label{v1}
\end{align}
On the other hand from the last condition of \eqref{conditions} we have
\begin{equation}
\nu=\frac{1}{1+\sum_j\beta_jP_j}\sum_{j=1}^K \beta_jP_j \left(\frac{w_j}{\frac{1}{a_j}+\frac{\beta_jP_j}{1+\sum_{j'} \beta_j' P_j'}}\right).\label{v2}
\end{equation}

Therefore for both \eqref{v1} and \eqref{v2} to hold the condition is
\begin{equation}\label{v3}
\frac{1}{1+\sum_j\beta_jP_j}\sum_{j=1}^K \beta_jP_j \left(\frac{w_j}{\frac{1}{a_j}+\frac{\beta_jP_j}{1+\sum_{j'} \beta_{j'} P_{j'}}}\right)\leq \min_k \frac{w_k}{\frac{1}{a_k}+x_k}.
\end{equation}
In such case we can always find the $\nu$ that satisfies the KKT conditions, which means that it is optimal to let every user use full power.
On the other hand if \eqref{v3} does not hold, there is no $\nu$ that can satisfy all conditions simultaneously, and we can easily check that $f(s)>0$.

When $s=1$, $\lambda_k=0, ~\forall k$ and therefore $f(s)=-\nu\leq0$.
Hence we have verified that $f(s)$ have different signs on the boundaries. Moreover \eqref{sumratedata2} is a convex problem and Slater's condition is satisfied. Therefore the KKT conditions are sufficient and necessary for optimality, there will be one $s$ such that $f(s)=0$ within the boundaries. The optimal $s$ can therefore be found via bisection.

\subsection{Proof of Lemma \ref{correlatedse}}\label{proofcorr}
We apply the bounding techniques used in \cite{JAMV2011} to develop an achievable SE as
\begin{equation}
\left(1-\frac{\tau}{T}\right)\log_2\left(1+ \mathrm{SINR}_k^{corr}\right),
\end{equation}
where
\begin{equation}
\mathrm{SINR}_k^{corr}=\frac{|E[\hat{\bg}_k^H\bg_k]|^2p_d^k}{\sum_j E[|\hat{\bg}_k^H\bg_j|^2]p_d^j-|E[\hat{\bg}_k^H\bg_k]|^2p_d^k+E[||\hat{\bg}_k||^2]}
\end{equation}
The results follow from calculating the terms in the above expressions, using standard results.
\color{black}
\end{appendices}

\bibliographystyle{IEEEtran}

\end{document}